\documentclass[a4paper]{article}

\usepackage[T1]{fontenc}
\usepackage[utf8]{inputenc}
\usepackage{fullpage}
\usepackage{authblk}
\usepackage{amsmath}
\usepackage{amssymb}
\usepackage{amsthm}
\usepackage{thmtools}
\usepackage{stmaryrd}
\usepackage{enumerate}
\usepackage{multicol}
\usepackage{todonotes}

\usepackage{tikz,tikz-qtree}
\usetikzlibrary{arrows,arrows.meta,automata,backgrounds,calc,decorations.pathmorphing,fit,positioning,shapes,snakes,trees}

\usepackage{hyperref} 
\hypersetup{
    colorlinks=false,
    linkcolor=blue,
    urlcolor=cyan,
}
\usepackage{cleveref}

\declaretheorem{theorem}
\declaretheorem{proposition}
\declaretheorem{lemma}

\newcommand{\N}{\mathbb{N}}
\newcommand{\G}{\mathcal{G}}
\newcommand{\R}{\mathsf{R}}
\newcommand{\prefix}{\mathsf{prefix}}
\newcommand{\suffix}{\mathsf{suffix}}
\newcommand{\val}{\mathsf{val}}
\newcommand{\per}{\mathsf{per}}

\newcommand{\bottomnode}{\mathsf{bottom}}
\DeclareMathOperator{\poly}{poly}
\newcommand{\rank}{\mathsf{rank}}
\newcommand{\select}{\mathsf{select}}
\newcommand{\sa}{\mathsf{sa}}
\newcommand{\isa}{\mathsf{isa}}
\newcommand{\weight}{\mathsf{wt}}

\title{Pattern Matching on Grammar-Compressed Strings \\ in Linear Time}

\author[1]{Moses Ganardi}
\author[2]{Pawe\l~Gawrychowski}

\affil[1]{Max Planck Institute for Software Systems (MPI-SWS), Germany}
\affil[2]{University of Wroc{\l}aw, Poland}

\date{}
\date{\vspace{-.5cm}}

\begin{document}
\clearpage\maketitle\thispagestyle{empty}

\begin{abstract}
The most fundamental problem considered in algorithms for text processing is pattern matching: given a pattern
$p$ of length $m$ and a text $t$ of length $n$, does $p$ occur in $t$? Multiple versions of this basic question
have been considered, and by now we know algorithms that are fast both in practice and in theory. However,
the rapid increase in the amount of generated and stored data brings the need of designing algorithms that
operate directly on compressed representations of data. In the compressed pattern matching problem we are given
a compressed representation of the text, with $n$ being the length of the compressed representation and $N$
being the length of the text, and an uncompressed pattern of length $m$. The most challenging (and yet
relevant when working with highly repetitive data, say biological information) scenario is when the chosen
compression method is capable of describing a string of exponential length (in the size of its representation).
An elegant formalism for such a compression method is that of straight-line programs, which are simply
context-free grammars describing exactly one string. While it has been known that compressed pattern
matching problem can be solved in $O(m+n\log N)$ time for this compression method, designing a linear-time
algorithm remained open. We resolve this open question by presenting
an $O(n+m)$ time algorithm that, given a context-free grammar of size $n$ that produces a single string $t$
and a pattern $p$ of length $m$, decides whether $p$ occurs in $t$ as a substring.
To this end, we devise improved solutions for the weighted ancestor problem and the substring concatenation problem.
\end{abstract}

\newpage
\setcounter{page}{1}

\section{Introduction}

In this paper, a text is simply a sequence of characters over some finite alphabet, sometimes called a string.
A canonical example is a DNA sequence, which is a sequence of characters over $\{\texttt{A},\texttt{C},\texttt{G},\texttt{T}\}$.
The most fundamental computational question considered in the area on algorithms for text processing is pattern
matching: given a pattern $p$ of length $m$ and a text $t$ of length $n$, does $p$ occur in $t$?
In the most basic version, we seek exact occurrences, that is, contiguous fragments of $t$ equal to $p$.
This can be solved in $O(n+m)$ time by the classical Knuth-Morris-Pratt algorithm \cite{KnuthMP77}, but
multiple other algorithms have been designed, e.g.~\cite{FaroL13} lists over 50 different algorithms
published after 2000, and mentions that almost 40 algorithms have been presented earlier.
Thus, by now the exact pattern matching seems to be well-understood, and we have solutions that
are efficient both in theory and in practice.

However, large datasets are rarely stored in an uncompressed form. This is particularly the case
with biological data, which is very often rather repetitive. Among many families of compression methods,
the most interesting from a theoretical point of view are those that allow for an exponential decrease
in the size of the compressed representation, such as the Lempel-Ziv compression or the related grammar compression.
The Lempel-Ziv family of compression methods consists of multiple specific algorithms, but on a very high
level they are all based on partitioning the text into blocks, with each block being defined using
the already encoded prefix of the text. The related grammar compression has a particularly clean
definition: the text is described with a context-free grammar describing exactly one string (that is,
every nonterminal appears exactly once on the left side, and the right side of the corresponding production
contains only terminal symbols and nonterminals with larger indices). Such a grammar is often
called a {\em straight-line program}, or SLP for short. Without losing the generality, the grammar is assumed
to be in {\em Chomsky normal form}, i.e. all rules are of the form $A \to BC$ or $A \to a$ where $A,B,C$ are nonterminals and $a$ is a terminal symbol.

The abundance of large datasets stored in a compressed forms raises the challenge of designing
algorithms that operate directly on the compressed representation, without explicitly decompressing the
whole input. This brings us to the problem considered in this paper: compressed pattern matching.
In this problem, we are given a compressed representation of a text $t$ of length $N$, with $n$
being the size of the compressed representation, and an uncompressed pattern $p$ of length $m$,
and should decide if $p$ occurs in $t$.

This question has received quite a bit of attention in the past. Amir, Benson, and Farach~\cite{AmirBF96}
considered its complexity for the Lempel-Ziv-Welch compression method (a simpler and less powerful
variant of the general Lempel-Ziv method), and designed two algorithms with running time
$O(n\log m+m)$ and $O(n+m^{2})$. The latter has been soon improved to $O(n+m^{1+\epsilon})$~\cite{Kosaraju95}.
For the general Lempel-Ziv compression method (more specifically, the so-called LZ77),
Farach and Thorup~\cite{FarachT98} designed an $O(n\log^{2}(N/n)+m)$ time algorithm.
Later, Gawrychowski obtained a clean $O(n+m)$ time algorithms for the Lempel-Ziv-Welch
compression method~\cite{Gawrychowski13}, and improved the complexity for the LZ77
compression method to $O(n\log(N/n)+m)$~\cite{Gawrychowski11}.

The high-level idea of the algorithm of Gawrychowski~\cite{Gawrychowski11} for pattern matching
in LZ77 compressed text is as follows. By a result of Charikar et al.~\cite{CharikarLLPPSS05},
a Lempel-Ziv parse of size $n$ can be converted into a balanced SLP of size $O(n\log(N/n))$. This means that,
for every production $A\rightarrow BC$, we have $\frac{\alpha}{1-\alpha}\leq \frac{|B|}{|C|} \leq \frac{1-\alpha}{\alpha}$,
for some constant $0 < \alpha \le 1/2$,
where $|X|$ denotes the length of the (unique) string derived by $X$.
Then, using the fact that the grammar is balanced, we can detect an occurrence of the pattern
in total time $O(n\log(N/n)+m)$.

Arguably, the main objective in the area of exact pattern matching is to achieve clean linear
time complexity.
For self-referential Lempel-Ziv compression, we know that unless one allows constant-time integer
division $\Omega(n\log N+m)$ operations are necessary. However, this lower bound does not
apply for non self-referential Lempel-Ziv compression, and in particular does not exclude
the possibility of a $O(n+m)$ time algorithm for grammar compression.
However, the best upper bound for the case of grammar compression
was the same as for the general Lempel-Ziv compression,
that is, either $O(n\log N+m)$ when we do not allow constant-time integer division
or $O(n\log(N/n)+m)$ if we do.

Recently, Ganardi, Jeż, and Lohrey~\cite{GanardiJL19} showed how to transform in linear time
an SLP of size $n$ describing a string $s$ of length $N$ into an equivalent SLP of size $O(n)$ with derivation tree of depth $O(\log N)$.
Thus, now we can assume without losing generality that the given grammar has depth $O(\log N)$.
This makes it particularly simple to, say, implement random access in linear space and logarithmic time,
significantly simplifying the previously known technically nontrivial result of Bille et al.~\cite{BilleLRSSW15}.
Clearly, a balanced grammar has depth $O(\log N)$ but not vice versa. Nevertheless, this exciting
progress suggests that one should revisit the complexity of pattern matching on grammar-compressed
strings, and seek a linear-time algorithm for grammars with logarithmic depth, which would then
imply a clean linear-time algorithm for any grammar.

\paragraph{Our result.}
In this paper, we successfully tackle the challenge of designing a linear-time algorithm for
pattern matching on grammar-compressed strings, and prove the following.

\begin{theorem}
	Given a pattern $p$ of length $m$, and an SLP $\G$ of size $n$,
	we can decide whether $p$ occurs in the text described by $\G$ in time $O(n+m)$.
\end{theorem}

\noindent In the above theorem and the whole paper we assume the standard word RAM model, which operates on $w$-bit words,
where $w \ge \log N$ and $w \ge \log m$, with the standard arithmetic (excluding integer division) and bitwise operations.

\paragraph{Techniques and comparison with prior work.}

The first step in our solution is to apply the result of Ganardi et al.~\cite{GanardiJL19} to make
the depth of the grammar $O(\log N)$. Then, if there is an occurrence of the pattern $p[1..m]$ then
there exist a production $A\rightarrow BC$ such that $p[1..i]$ is a suffix of the string described
by $B$ while $p[(i+1)..m]$ is a prefix of the string described by $C$. Thus, the natural
approach is to check, for each nonterminal $A$ of $\mathcal{G}$, whether the string it describes is a substring of $p$,
and if not compute its longest prefix that is a suffix of $p$ and the longest suffix that is a prefix of $p$.
This was the approach taken in~\cite{Gawrychowski11}. With some insight related to combinatorics
on words, such information is enough to detect an occurrence in constant time per production,
see~\cite[Lemma 7]{Gawrychowski11}. However, computing the information bottom-up for each production separately
seems to require logarithmic time per nonterminal. This difficulty was overcome in~\cite{Gawrychowski11}
by processing multiple productions together, more specifically by batching together nonterminals
deriving strings of roughly the same length (up to constant factors), and computing just some approximation
of this information for each nonterminal.
An important property of a balanced grammar is that, after splitting the nonterminals into such layers, productions
for all nonterminals in the same layer refer to the nonterminals in a constant number of previous layers.
This was the key insight that allowed for processing all nonterminals in $O(|\mathcal{G}|+m)$ total time.
However, the balancing technique of Ganardi et al.~\cite{GanardiJL19} only guarantees that
the depth of $\mathcal{G}$ is logarithmic, which is not enough for such an approach to work%
\footnote{A recent result of Ganardi~\cite{ganardi2021compression} guarantees that the depth of every subtree
is logarithmic in the length of the derived string, but this is also not enough.}.
In fact, it was shown in~\cite{ganardi2021compression} that
any transformation of arbitrary SLPs into balanced SLPs (in the sense of Charikar et al.~\cite{CharikarLLPPSS05})
must incur a multiplicative blowup of $O(\log N)$ where $N$ is the string length.
Thus, we need to design a new algorithm.

Our improved solution is based on extending the combinatorial insight used in the prior work
and combining it with appropriate data structures. For the data structures part, we work with
the substring concatenation problem, which asks for preprocessing the pattern $p[1..m]$ to allow
for checking if the concatenation of any two of its substrings $p[i..j]p[i'..j']$ occurs in the whole $p$.
This is a basic building block in other algorithms, e.g. Amir et al.~\cite{AmirLLS07} designed
an $O(m\sqrt{\log m})$ space structure with $O(\log\log m)$ query time to solve some problems
on dynamic texts. Using a linear-space constant-time data structure for the so-called weighted
ancestor problem by Gawrychowski et al.~\cite{GawrychowskiLN14}, Bille et al.~\cite{BilleCCGSVV18}
obtained improved space-time tradeoffs for this problem. However, in this particular
application we would need a linear-space constant-time data structure that can be constructed
in linear time. Even though the very recent result of Belazzougui et al.~\cite{BelazzouguiKPR21}
does provide such a data structure for the weighted ancestor problem, it is not clear how
to extend it to the substring concatenation problem with the same time and space bounds.
Thus, we take another approach, and exploit the fact that in this case we can afford to
batch multiple queries together.

We present improved offline algorithms for the {\em weighted ancestor problem} and the {\em substring concatenation problem}.
In the weighted ancestor problem we are given a node-weighted tree.
The weights are nonnegative $w$-bit integers and strictly increasing on a path from the root to any node,
i.e. the weight of a node is greater than the weight of its parent.
A weighted ancestor query asks: Given a node $u$ and a number $k \in \N$,
return the furthest ancestor of $u$ with weight at least $k$.
We can assume that $u$ is a leaf since we can store pointers from every node to a descendant leaf.
Very recently, it was show that weighted ancestor queries on suffix trees can be answered in constant time
after linear time and space preprocessing~\cite{BelazzouguiKPR21},
which allows to find the node of a substring $u[i..j]$ in constant time.
We present another (simpler) offline solution that builds on the result
by Kociumaka et al.~\cite{KociumakaKRRW20}, who showed how to perform $q$ weighted ancestor queries in $O(q+s)$ time
on a general tree of size $s$, assuming that the queries are sorted by their weights.
In our application, we need to replace $s$ with (at most) $s/\log N$ in the time complexity.

\begin{restatable}{theorem}{weightedancestors}
	\label{thm:weightedancestors}
	A tree $T$ of size $s$ and weights up to $m$ can be preprocessed in $O(s)$ time
	so that $q$ weighted ancestor queries can be answered in time $O(q + s / w)$
	and one call to sorting $q$ integers up to $m$.
\end{restatable}

A substring concatenation query on a string $p$ asks:
Given two substrings $u = p[i..j]$ and $v = p[k..\ell]$ of $p$, check whether $uv$ is a substring of $p$
and, if so, return the position of an occurrence.
We are not aware of a previous offline solution for this problem. In our application
it is crucial that the time is linear in the number of queries and sublinear in the length of the pattern.

\begin{restatable}{theorem}{substringconcat}
	\label{thm:substringconcat}
	The pattern $p$ of length $m$ can be preprocessed in $O(m)$ time
	so that $q$ substring concatenations can be answered in time $O(q + m / w)$.
\end{restatable}

\paragraph{Organisation of the paper.} 
We start with the preliminaries in \Cref{sec:preliminaries}.
We postpone the proofs of \Cref{thm:weightedancestors} and \Cref{thm:substringconcat}
to \Cref{sec:weightedancestors} and \Cref{sec:substringconcat}, and assume them as already
proved in \Cref{sec:main}, where we present the main algorithm.

\paragraph{Related work.}

In the {\em fully compressed pattern matching problem}, both the text and the pattern are given by straight-line programs.
This problem is known to be solvable in polynomial-time \cite{GasieniecKPR96,Jez15,KarpinskiRS95,Lifshits07,MiyazakiST97}
and the currently fastest solution is due to Je\.z \cite{Jez15}, with a running time of $O((n + m) \log M)$
where $n,m$ are the sizes of the given SLPs for the text and the pattern, respectively, and $M$ is the pattern length.
The latter solution uses the {\em recompression} technique, which has also been applied to
compressed membership problems for finite automata \cite{Jez14}, word equations \cite{Jez16}, equations in free groups \cite{DiekertJP16}, and context unification \cite{Jez19}.
Gąsieniec and Rytter also presented an $O((n+m) \log (n+m))$ time solution for the fully compressed pattern matching problem
for LZW-compressed strings \cite{GasieniecR99}, which was later improved to linear time by Gawrychowski \cite{Gawrychowski12}.
A closely related topic is the {\em compressed text indexing} problem where an index is a data structure
that supports efficient pattern matching queries on the text.
A good overview of recent results on compressed indices can be found in the excellent survey by Navarro~\cite{Navarro21}.

\section{Preliminaries}
\label{sec:preliminaries}

We write $[i..j]$ for $\{i, \dots, j\}$ and $[n]$ for $\{1, \dots, n\}$.
For a string $s = a_1 \dots a_n$ we write $s[i] = a_i$ for the $i$-th character.
A {\em substring} of a string $u$ is a pair $(i,j)$ where $1 \le i \le j \le |u|$ and is identified with the string $u[i..j] = u[i] u[i+1] \dots u[j]$.
We say that {\em $u$ occurs in $v$ at position $i$} if $u = v[i+1..i+|u|]$
\footnote{This definition of an occurrence at position $i$ simplifies formulas throughout the paper.}.
A {\em period} of a string $u$ is an integer $d \ge 1$ with $u[i] = u[i+d]$ for all $1 \le i \le |u|-d$.
The smallest period $\per(u)$ of $u$ is also called {\em the period} of $u$.
If $d = \per(u) \le |u|/2$ then the periodicity lemma \cite{fine1965uniqueness} implies that the set of all periods $\le |u|/2$
forms an arithmetic progression $\{ \alpha d \mid \alpha \ge 1 \} \cap [0..|u|/2]$.

The {\em compacted trie} $T$ of a set of strings $S$ is obtained from the trie of $S$
by contracting unary paths.
The nodes in $T$ are also called {\em explicit nodes},
whereas {\em implicit nodes} are positions on an edge label.
The {\em string depth} of a (explicit or implicit) node $v$ in $T$ is the length of the string labelling
the path from the root to $v$.
The {\em suffix tree} of a word $u$ is a compacted trie of all suffixes of $u \, \$$
where $\$$ is a fresh symbol.
Later in \Cref{sec:substringconcat} we will also consider compacted tries only containing some suffixes.

In this paper we always denote by $p$ the pattern of length $m$.
In all algorithms we assume the following data structures on the pattern.
In $O(m)$ time we build the suffix trees for $p$ and $p^\R$ \cite{Ukkonen95}.
We label every explicit node by its string depth.
Furthermore, by traversing all leaves we compute in linear time an array of length $m$ which maps a number $i$ to the leaf corresponding to the suffix $u[i..|u|]$.
We preprocess the suffix trees in linear time such that they support {\em least common ancestor} queries in constant time \cite{BenderF00}.
This allows us to compute longest common prefixes of substrings (lcp queries) in constant time.
By a depth first traversal of the suffix tree we also compute in $O(m)$ time
the {\em suffix array} $\sa[1..m]$ of $p$ and
the {\em inverse suffix array} $\isa[1..m]$ of $p$ where
$p[\sa[1]..m], p[\sa[2]..m], \dots, p[\sa[m]..m]$ is the lexicographically ordered list of suffixes of $p$,
and $\isa[i]$ is the lexicographic rank of $p[i..m]$ in the set of all suffixes of $p$ (position in this ordering).
Using the preprocessing of the Knuth-Morris-Pratt algorithm we can compute the periods of all prefixes of $p$
in linear time, see e.g. \cite[Lemma~3.3]{CrochemoreR94}.

\begin{lemma}
	\label{lem:prec-periods}
	One can compute the periods of all prefixes and suffixes of $p$ in $O(m)$ time.
\end{lemma}

Let $\prefix(u)$ be the longest prefix of $u$ which is a suffix of $p$,
and let $\suffix(u)$ be the longest suffix of $u$ which is a prefix of $p$.

\begin{lemma}
	\label{lem:pref-suf}
	The pattern $p$ can be preprocessed in $O(m)$ time, such that
	given substrings $u_1, \dots, u_q$ of $p$, one can compute $\prefix(u_i)$ and $\suffix(u_i)$ for all $1 \le i \le q$ in time $O(q + m/w)$.
\end{lemma}

\begin{proof}
	Every substring $u_i$ corresponds to a (possibly implicit) node in the suffix tree.
	These nodes can be computed in time $O(q + m/w)$ using the weighted ancestor data structure.
	For every explicit node we precompute its nearest ancestor with a $\$$-labeled child in time $O(m)$,
	which corresponds to the $\prefix$-information.
	Similarly, the $\suffix$-information can be computed from the suffix tree of the reversed pattern.
\end{proof}

A {\em straight-line program (SLP)} is a context-free grammar $\G$ such that
(i) every nonterminal occurs exactly once on the left-hand side of a rule
and (ii) there exists a linear order $<$ on the nonterminals such that $A < B$ whenever $B$ occurs
on the right-hand side of a rule $A \to u$.
This ensures that every nonterminal $A$ derives a unique terminal string $\val(A)$,
and we set $\val(\G) = \val(S)$ where $S$ is the start nonterminal.
The {\em size} $|\G|$ of $\G$ is the total length of the right-hand sides of all rules.
We can always assume that all nonterminals and rules are reachable from the start variable
and that $\G$ is in {\em Chomsky normal form},
i.e. all rules are of the form $A \to BC$ or $A \to a$ where $A,B,C$ are nonterminals and $a$ is a terminal symbol.
Furthermore, by \cite{GanardiJL19} we can transform $\G$ into an SLP of size $O(n)$
whose derivation tree has height $O(\log N)$
in $O(n)$ time where $N = |\val(\G)|$.
We refer to \cite{Lohrey12} for a good overview on grammar-based compression.

\section{Reduction to substring concatenation}
\label{sec:main}

Consider an SLP $\G$ of size $n$ for a text of length $N$
and a pattern $p$ of length $m \ge 2$.
We define $\prefix(A) = \prefix(\val(A))$ and $\suffix(A) = \suffix(\val(A))$ for nonterminals $A$.
Observe that the pattern $p$ occurs in $\val(\G)$ if and only if there exists a rule $A \to BC$
such that $p$ occurs in $\suffix(B) \, \prefix(C)$.
Instead of computing $\prefix(A)$ and $\suffix(A)$
we will compute the following approximation for every nonterminal $A$ in $\G$:
\begin{enumerate}
\item If $\val(A)$ occurs in $p$ we compute the {\em substring information} for $A$,
i.e. a substring $s_A$ of $p$ with $\val(A) = s_A$.
\item If $\val(A)$ does not occur in $p$ we compute the {\em prefix} and the {\em suffix information} for $A$,
i.e. two substrings $x_A$ and $y_A$ such that $\prefix(A)$ is a prefix of $x_A y_A$ which in turn is a prefix of $\val(A)$, and
two substrings $u_A$ and $v_A$ such that $\suffix(A)$ is a suffix of $u_A v_A$ which in turn is a suffix of $\val(A)$.
\end{enumerate}

\begin{proposition}
	\label{prop:comp-info}
	One can compute the information above in time $O(n + m)$.
\end{proposition}

\begin{proof}
By \cite{GanardiJL19} we can restructure $\G$ so that $\G$ has size $O(n)$ and the derivation tree of $\G$ has height $O(\log N)$.
Furthermore, all nonterminals of the original SLP are present in the new SLP, deriving the same strings.
Let $L_k$ be the set of nonterminals $A$ in $\G$ whose derivation tree has height $k$.
The goal is to compute the information for all nonterminals in $L_k$, assuming the information
has been computed already for $L_0, \dots, L_{k-1}$ in time $O(|L_k| + m/w)$.
In total, this sums up to
$\sum_{k=1}^{O(\log N)} O\left(|L_k| + m/w\right) = O(n + m)$
since $w \ge \log N$.

Consider a rule $A \to BC$ where $A \in L_k$ and $B,C \in \bigcup_{i < k} L_i$.
If $\val(B) = s_B$ and $\val(C) = s_C$ we search for a concatenatation $s_B s_C$ in $p$.
If this is successful, we have the substring information for $A$.
Otherwise, $s_B s_C$ is both the prefix information and the suffix information for $A$
since it covers both $\prefix(A)$ and $\suffix(A)$.
If either $\val(B)$ or $\val(C)$ does not occur in $p$ then also their concatenation does not occur in $p$,
and we need to compute the prefix and suffix information for $A$.

Suppose that $\val(B) = s_B$ and $\prefix(C)$ is a prefix of $x_C y_C$ which in turn is a prefix of $\val(C)$.
Then $\prefix(A)$ is a prefix of $s_B x_C y_C$.
We search for an occurrence of $s_B x_C$ in $p$:
If we are successful we replace $s_B x_C$ by that substring.
Otherwise $\prefix(A)$ is a prefix of $s_B x_C$.
Furthermore, $\suffix(A) = \suffix(C)$ since otherwise $\val(C)$ would occur in $p$.
Similarly, we treat the case where $\val(B)$ does occur in $p$ but $\val(C)$ does not.

If both $\val(B)$ and $\val(C)$ do not occur in $p$ then neither does $A$.
Furthermore $\prefix(A) = \prefix(B)$ and $\prefix(A) = \suffix(C)$.

Notice that for every of the $O(\log N)$ layers $L_k$ we only need to solve a batch of $|L_k|$ queries
of the substring concatenation problem, taking $O(|L_k| + m/w)$ time using \Cref{thm:substringconcat}.
\end{proof}

\begin{figure}

\centering
\begin{tikzpicture}

\tikzset{every node/.style={rectangle, draw, inner sep = 0, minimum height = 1em}}
\tikzstyle{string}=[minimum width = 5em]
\tikzstyle{pattern}=[minimum width = 12em]

\node [string] at (0em,0em) {$u$};
\node [string] at (5em,0em) {$v$};
\node [string] at (10em,0em) {$x$};
\node [string] at (15em,0em) {$y$};

\node [pattern] at (5em,-1.5em) {$p$};
\node [pattern] at (7.5em,-3em) {$p$};
\node [pattern] at (10em,-4.5em) {$p$};

\end{tikzpicture}

\caption{Pattern matching in a concatenation of four substrings can be reduced to three substrings.}
\label{fig:four-to-three}

\end{figure}
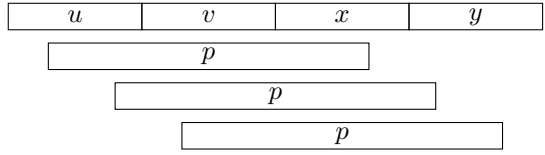

Hence for every nonterminal $A$ we can compute four (possibly empty) substrings $x_A, y_A, u_A, v_A$ of $p$,
so that $x_A y_A$ lies between $\prefix(A)$ and $\val(A)$ in the prefix ordering
and $u_A v_A$ lies between $\suffix(A)$ and $\val(A)$ in the suffix ordering.
Observe that $p$ occurs in $\val(\G)$ if and only if $p$ occurs in $u_B v_B x_C y_C$ for some rule $A \to BC$.
Hence we have reduced pattern matching to the following problem:
Given a set $Q$ of $O(n)$ many quadruples $(u,v,x,y)$ of substrings of $p$,
does $p$ occur in $uvxy$ for some tuple $(u,v,x,y) \in Q$?
We can reduce the number of substrings from four to three in time $O(n+m)$, see \Cref{fig:four-to-three}:
If $p$ occurs in $uvxy$ then it occurs in either $uvx$ or $vxy$,
or $vx$ occurs in $p$.
For all tuples $(u,v,x,y) \in Q$ we test whether $vx$ occurs in $p$ using a substring concatenation query.
If so, we replace $vx$ by a single substring (using \Cref{thm:substringconcat}),
and otherwise we replace the quadruple $(u,v,x,y)$
by the triples $(u,v,x)$ and $(v,x,y)$.
This yields a set $\hat Q$ of substring triples of size $O(n)$.
It remains to search for occurrences of $p$ in $uvx$ for some $(u,v,x) \in \hat Q$.
The following proposition generalizes \cite[Lemma~3.1]{Gawrychowski11a} and
\cite[Lemma~6]{Gawrychowski11}, respectively.

\begin{proposition}
	\label{prop:triples}
	Given a finite set $Q$ of substring triples $(u,v,x)$ of the pattern $p$,
	we can test in time $O(|Q|+m)$ whether $p$ occurs in $uvx$ for some $(u,v,x) \in Q$.
\end{proposition}

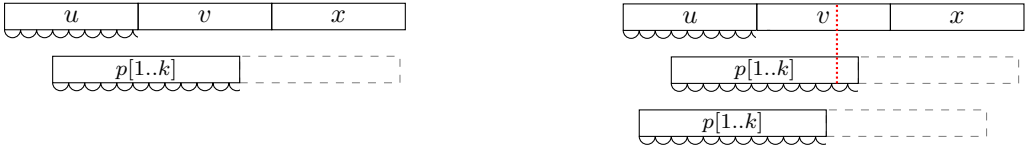
\begin{figure}

\def \combinatorics%
{
\tikzset{every node/.style={rectangle, draw, inner sep = 0, minimum height = 1em}}
\tikzstyle{string}=[minimum width = 5em]
\tikzstyle{period}=[snake=bumps,segment length=1.2em, segment amplitude = -0.3em]
\tikzstyle{erase}=[fill = white, draw = none]

\node [string] at (0em,0em) {$u$};
\draw[period] (-2.5em,-0.5em) -- (3em,-0.5em);
\draw[erase] (2.5em,-0.5em) rectangle (3em,-1em);

\node [string] at (5em,0em) {$v$};
\node [string] at (10em,0em) {$x$};

\draw[period] (-0.7em,-2.5em) -- (7em,-2.5em);
\draw[erase] (6.3em,-1.5em) rectangle (10em,-3em);
\node[minimum width = 7em] at (2.8em,-2em) {\footnotesize $p[1..k]$};
\node[minimum width = 6em, dashed, draw opacity = 0.5] at (9.3em,-2em) {};
}

\centering
\begin{multicols}{2}
\begin{tikzpicture}
\combinatorics
\end{tikzpicture}
\begin{tikzpicture}
\combinatorics
\draw[densely dotted, thick, red] (5.5em,-2.5em) -- (5.5em,0.5em);
\draw[period] (-1.9em,-4.5em) -- (5.8em,-4.5em);
\draw[erase] (5.1em,-4em) rectangle (8.8em,-5em);
\node[minimum width = 7em] at (1.6em,-4em) {\footnotesize $p[1..k]$};
\node[minimum width = 6em, dashed, draw opacity = 0.5] at (8.1em,-4em) {};
\end{tikzpicture}
\end{multicols}
\caption{Searching for an occurrence of $p$ in $uvx$ with a large overlap with $u$.}
\label{fig:overlap-ux}

\end{figure}

\begin{proof}

We will show how to test an occurrence of $p$ in $uvx$ for any triple $(u,v,x) \in Q$
using a constant number of weighted ancestor queries, substring concatenations, applications of \Cref{lem:pref-suf}
and other operations.
Since we can batch the queries into sets of $O(|Q|)$ queries
we can process the entire set $Q$ in time $O(|Q| + m)$.

Suppose we want to test whether $p$ occurs in the concatenation of given substrings $u,v,x$ of $p$.
Observe that $p$ occurs in $uvx$ if and only if it occurs in $\suffix(u) \, v \, \prefix(x)$.
By replacing $u$ by $\suffix(u)$ and $x$ by $\prefix(x)$
we can assume that $u$ is a prefix of $p$ and $x$ is a suffix of $p$.

\paragraph{Case 1: Large overlaps with $u$ or $x$.}
We can precompute the period $d$ of $u$ by \Cref{lem:prec-periods}.
We search for occurrences of $p$ in $uvx$ at positions $i \le |u|/2$.
Since $u$ is a prefix of $p$ such a position $i$ must be a period of $u$ (or $i = 0$)
and therefore it is of the form $i = \alpha d$ where $0 \le \alpha \le |u|/(2d)$.
Compute the maximal $k \ge |u|$ such that $p[1..k]$ has period $d$,
by computing the lcp between $p$ and $p[d+1..m]$.
Let $\alpha_{\max} \le |u|/(2d)$ be maximal such that $\alpha_{\max} d + m \le |uvx|$
(if there is no such $\alpha_{\max}$ then $p$ is longer than $uvx$ and does not occur).
Using at most two lcp queries we test whether $p[1..k]$ occurs in $uvx$ at position $\alpha_{\max} d$
and, if not, compute the leftmost mismatch.
In the following we compute an occurrence of $p$ in $uvx$ or eliminate all but one candidate position $\alpha d$.
In the latter case we test whether $p$ occurs at $\alpha d$ using at most two lcp queries.
\begin{itemize}
\item Assume that $p[1..k]$ occurs at position $\alpha_{\max} d$.
If $k = m$ then we have found an occurrence of $p$ in $uvx$.
If $k < m$ we claim that $p$ cannot occur at position $\alpha d$ where $\alpha < \alpha_{\max}$:
Since $p[1..k]$ occurs at position $\alpha_{\max} d$ we know that the prefix of $uvx$ of length $\alpha_{\max}d+k$
is $d$-periodic.
An occurrence at position $\alpha d$ would imply
\[
	p[k+1] = (uvx)[\alpha d + k+1] = (uvx)[\alpha_{\max} d + k+1 - d] = p[k+1-d],
\]
contradicting the fact that $p[1..k]$ is the maximal $d$-periodic prefix of $p$.
\item Assume that there is a mismatch, say $i \in [1..k]$ is minimal with
$p[i] \neq (uvx)[\alpha_{\max} d + i]$.
Observe that the prefix of $uvx$ of length $\alpha_{\max} d + i - 1$ is $d$-periodic.
We claim that an occurrence of $p[1..k]$ cannot cover the mismatch,
i.e. $p[1..k]$ cannot occur at positions $\alpha d$ with $\alpha d + k \ge \alpha_{\max} d + i$:
Otherwise $i + (\alpha_{\max} - \alpha)d \le k$ and thus
\[
	p[i] = p[i + (\alpha_{\max} - \alpha)d] = (uvx)[i + \alpha_{\max} d].
\]
Here the first equality uses that $p[1..k]$ is $d$-periodic and the second equality uses that
$p[1..k]$ occurs at position $\alpha d$.
This contradicts the assumption that $p[i] \neq (uvx)[\alpha_{\max} d + i]$.

Hence $p$ can only occur at positions $\alpha d < \alpha_{\max} d + i - k$.
If $k = m$ then $p$ occurs at any such position $\alpha d$ by $d$-periodicity of
the prefix of $uvx$ of length $\alpha_{\max} d + i - 1$.
If $k < m$ we claim that $p$ can only occur at the maximal position $\alpha d$
where $\alpha d + k < \alpha_{\max} d + i$.
Towards a contradiction, suppose that $p$ occurs at position $\alpha d$
where $(\alpha+1) d + k < \alpha_{\max} d + i$.
Then the prefix of $uvx$ of length $(\alpha+1) d + k$ is $d$-periodic, and thus
\[
	p[k + 1 - d] = (uvx)[\alpha d + k + 1 - d] = (uvx)[\alpha d + k + 1] = p[k+1],
\]
which contradicts the fact that $p[1..k]$ is the maximal $d$-periodic prefix.
\end{itemize}
If we have not found any occurrence we can replace $u$ by its suffix of length $\lfloor |u|/2 \rfloor$
and repeat the same procedure from above.
After at most three iterations we can ensure that $|u| \le m/4$.
By applying a symmetric argument to $x$ we can ensure that $|x| \le m/4$.
If $|v| < m/2$ then $|uvx| < m$ and $p$ does not occur in $uvx$.

\begin{figure}
\centering
\begin{tikzpicture}
\tikzset{every node/.style={rectangle, draw, inner sep = 0, minimum height = 1em}}
\tikzstyle{period}=[snake=bumps,segment length=1.2em, segment amplitude = -0.3em]
\tikzstyle{erase}=[fill = white, draw = none]

\draw[period] (-5.75em,-0.5em) -- (9em,-0.5em);
\draw[period] (-7em,-0.5em) -- (-5.75em,-0.5em);
\draw[erase] (7.5em,-0.5em) rectangle (10em,-2em);
\draw[erase] (-6.85em,-0.5em) rectangle (-9em,-2em);
\node [minimum width = 5em] at (-8.25em,0em) {};
\node [minimum width = 11.5em] at (0em,0em) {$v$};
\node [minimum width = 4em] at (7.75em,0em) {};
\node [minimum width = 1.75em] at (6.625em,0em) {\scriptsize $x^{}_1$};
\node [minimum width = 2.25em] at (8.63em,0em) {\scriptsize $x^{}_2$};

\node [minimum width = 14.35em] at (0.325em,-2em) {$s$};
\draw[period] (-5.75em,-2.5em) -- (9em,-2.5em);
\draw[erase] (7.5em,-2em) rectangle (10em,-3em);
\node [minimum width = 1.5em] at (8.25em,-2em) {$t$};
\draw[period] (-7em,-2.5em) -- (-5.75em,-2.5em);
\draw[erase] (-6.85em,-2em) rectangle (-9em,-3em);
\node [minimum width = 1.3em] at (-7.5em,-2em) {$r$};

\node [minimum width = 3.9em] at (-8.8em,0em) {\scriptsize $u^{}_1$};
\node [minimum width = 1.1em] at (-6.3em,0em) {\scriptsize $u^{}_2$};
\end{tikzpicture}
\caption{Searching for an occurrence of $p$ in $uvx$ where $v$ is long.}
\label{fig:long-v}

\end{figure}
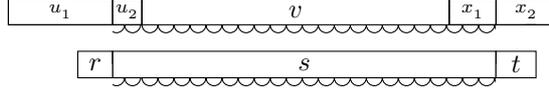

\paragraph{Case 2: $v$ is long.}
Now we assume that $|v| \ge m/2$ and $|u|,|x| \le m/4$.
We can again assume that $u$ and $x$ are a prefix and a suffix of $p$.
We can detect whether $p$ occurs in $uv$ (or in $vx$) with the same argument as in the case above
since such an occurrence must cover the suffix (prefix, respectively) of $v$ of length $3m/4 \ge |v|/2$.
If both of these tests are negative, any occurrence of $p$ in $uvx$ must cover the substring $v$.
We locate the node of the substring $v$ in the suffix tree.
The leaves below that node identify the occurrences of $v$ in $p$.
If there is only occurrence of $v$, say $v = p[i..j]$,
then we test whether it extends to an occurrence of $p$,
i.e. $p[1..i-1]$ is a suffix of $u$ and $p[j+1..m]$ is a prefix of $x$, using two lcp queries.

Now assume that there are at least two occurrences of $v$ in $p$.
If $p$ occurs in $uvx$ then every occurrence of $v$ in $p$ covers at least $|v| - \max\{|u|,|x|\} \ge |v| - m/4$ symbols of the explicit substring $v$ in $uvx$.
Since $v$ occurs at least twice in $p$ there must be such an occurrence of $v$ in $p$
that has a proper overlap with the explicit substring $v$ in $uvx$.
Hence, if $p$ occurs in $uvx$ then $\per(v) \le m/4 \le |v|/2$.
We can (pre)compute the difference $d$ of the positions of the first two occurrences of $v$.
If $d > |v|/2$ then $p$ does not occur in $uvx$, so we can assume that $d = \per(v) \le |v|/2$.
Consider the maximal substring $s$ of $p$ which is the periodic extension of an arbitrary occurrence of $v$ in $p$,
i.e. $\per(s) = \per(v)$, see~\Cref{fig:long-v}.
We can compute the factorization into substrings $p = rst$,
by starting with any occurrence of $v$ in $p$ and extending it to the left and to the right using two lcp queries.
Similarly, we compute how far the period of $v$ extends into $u$ and $x$, using two lcp queries.
We factor $u = u_1u_2$ and $x = x_1 x_2$ such that $u_2 v x_1$ is the maximal substring of $uvx$ with period $d$.

Whenever $p = rst$ occurs at position $i$ in $u_1 u_2 v x_1 x_2$ then $s$ occurs in $u_2 v x_1$ at position $i + |r| - |u_1|$.
We can compute all occurrences of $s$ in $u_2 v x_1$ as follows:
Since $v$ is a substring of $s$ we can compute an occurrence $k$ of $v$ in $s$, i.e. $v = s[k+1..k+|v|]$.
Then the occurrences of $s$ in $u_2 v x_1$ are the positions $|u_2| - k + \alpha d$ in the interval $[0..|u_2 v x_1|-s]$
where $\alpha \in \mathbb{Z}$.
If there is no such occurrence then $p = rst$ also does not occur in $u_1 u_2 v x_1 x_2 = uvx$.
Now assume that $s$ does occur in $u_2 v x_1$.
\begin{itemize}
\item If $r$ and $t$ are empty then $p = s$ occurs in $uvx$.
\item Suppose that $r$ is nonempty. We claim that if $p$ occurs at position $i$ in $uvx$
then $j = i + |r| - |u_1|$ must be the leftmost occurrence of $s$ in $u_2 v x_1$.
Towards a contradiction, suppose that $s$ also occurs at position $j - d$ in $u_2 v x_1$.
Then we have
\[
	p[|r|] = (uvx)[i + |r|] = (uvx)[j + |u_1|] = (u_2 v x_2)[j] = (u_2 v x_2)[j-d+d] = s[d] = p[|r|+d],
\]
which contradicts the maximality of $s$ in $p$.
Hence we can compute the leftmost occurrence $j$ of $s$ in $u_2 v x_1$,
yielding a candidate position $i = j + |r| - |u_1|$ for $p$ at $uvx$,
which can be verified using a constant number of lcp queries.
\item If $t$ is nonempty we proceed symmetrically using the rightmost occurrence of $s$ in $u_2 v x_1$.
\end{itemize}
This concludes the proof.
\end{proof}

\section{Weighted ancestor queries}

\label{sec:weightedancestors}

In this section, we will work with arbitrary node-weighted trees. We assume that the weights are strictly increasing
on each root-to-leaf path, and denote by $\weight(u)$ the weight of a node $u\in T$. A weighted ancestor query asks,
given a node $u\in T$ and a weight $k \leq \weight(u)$, to find the furthest ancestor $v$ of $u$ such that $\weight(v)\geq k$.
The answer to such a query does not change if we replace $u$ by any leaf in its subtree, thus by storing for each $u$
a pointer to any leaf in its subtree it is enough to show how to answer a weighted ancestor query for $u$ being a leaf.

We will extensively use the following decomposition of an tree $T$ on $n$ leaves and a parameter $x$
(similar to the ART-decomposition of Alstrup et al.~\cite{ART}), or $x$-decomposition for short. First, we order the children of every node of $T$
(any ordering suffices) and order all nodes according to their numbers in the preorder traversal.
We choose
every $x^{\text{th}}$ node of $T$ in this order and define the top tree $T'$ to be the subtree of $T$ induced by the 
root, all chosen nodes, and the least common ancestor of every two chosen nodes.
The parent node of a non-root node $v$ in $T'$ is the nearest ancestor of $u$ in $T$ that also appears in $T'$
(the root of $T$ becomes the root of $T'$).
The only nodes with one child in $T'$ are possibly the root and some of the chosen nodes.
Thus, the total number of nodes in $T'$ is $O(n/x)$.
If we remove from $T$ all paths from the root to the chosen nodes,
we obtain a set of subtrees of $T$, called {\em bottom trees}.
Notice that a bottom tree is rooted at a child of a node
that belongs to a path from the root to some chosen node.
Further, the nodes of each bottom tree constitute a contiguous fragment in the preorder traversal,
and hence their size is less than $x$.

\subsection{Predecessor queries and sorting}

Let $S$ be a finite ordered set of elements, say a set of integers or of words ordered lexicographically.
For an element $x$ we define $\rank(S,x) = |\{y \in S \mid x < y\}|$.
For a number $i$ we define $\select(S,i) = x$ such that $\rank(S,x) = i$.

\begin{lemma}[\cite{PatrascuT14}]
\label{lem:predecessor}
Given a set $S$ of $s \le \poly(w)$ integers, each consisting of $w$ bits, we can construct
in $O(s)$ time and space a structure so that we can compute $\rank(S,x)$ and $\select(S,i)$ in constant time.
\end{lemma}

\begin{lemma}
\label{lem:sorting}
A set of $n$ integers from $[m]$ can be sorted in $O(n+m/w)$ time.
\end{lemma}

\begin{proof}
For $w\leq\sqrt{n+m}$ we use radix sort to sort in $O(n+\sqrt{m})=O(n+\sqrt{n+m})=O(n+(n+m)/w)=O(n+m/w)$ time.
For $w>\sqrt{n+m}$ we proceed as in the proof of \Cref{lem:predecessor}, and observe that
the size of the maintained set is $n\leq w^{2}$. Then
we go over the input set and query the built structure to obtain the rank of each integer
in the set of distinct integers, and then sort by counting in $O(n)$ time.
\end{proof}

\begin{lemma}
\label{lem:batchedpred}
Given a sorted list $S$ of $n$ integers, each consisting of $w$ bits, we can construct
in $O(n)$ time and space a structure that, given a sorted list of $q$ integers $x_{1},\ldots,x_{q}$, each consisting of $w$ bits,
computes $\rank(S,x_j)$ for each $j \in \{1, \dots, q\}$ in $O(q+n/w)$ total time.
\end{lemma}

\begin{proof}
We partition $S$ into $n/w$ blocks of size $w$. We separately store a sorted list containing the first
element from each block, and for each block we store its elements in a structure implemented with \Cref{lem:predecessor}.
To answer a query, we first merge in $O(q+n/w)$ time the sorted list of $x_{1},\ldots,x_{q}$ with the sorted list
containing the first element from each block. This gives us, for every $x_{j}$, a unique block where we should
search for its predecessor. We query the predecessor structure of the block to obtain the rank of the predecessor in the block,
which is then used to retrieve the rank of the predecessor in $S$.
\end{proof}

\subsection{Batched weighted ancestor queries}

To prove \Cref{thm:weightedancestors} we combine two solutions for the weighted ancestor problem:
By \cite[Lemma~7.2]{KociumakaKRRW20} we can answer $q$ weighted ancestor queries in $O(q+s)$ time,
assuming that the queries are sorted by their weights.
Furthermore, we use the following solution on small trees:

\begin{lemma}
\label{lem:micro-wa}
A tree of size $s \le O(w)$ can be preprocessed in $O(s)$ time so that we can answer weighted ancestor queries online in constant time.
\end{lemma}

\begin{proof}
First we transform $T$ into a tree $\tilde T$ with pairwise distinct weighted depths:
Let $v_1, \dots, v_s$ be a depth-first traversal of $T'$.
We replace the weight $\weight(v_i)$ of a node $v_i$ by $\weight(v_i) \cdot 2^w + i$.
We remark that all standard operations on a $2w$-bit word RAM can be simulated by a constant number of $w$-bit operations.
We store all node weights in $\tilde T$ in a predecessor data structure $\tilde V$ from \Cref{lem:predecessor},
supporting constant time rank and select queries.
Additionally, we store in each node $v$ of $\tilde T$ a bitvector $b(v)$ of length $s \le O(w)$ (with 0-based indexing)
whose $i$-th bit is one if and only if $\select(\tilde V, i)$ is the weight of an ancestor of $v$.
These bitvectors can be computed in linear time:
The bitvector $b(v)$ can be obtained from the bitvector of its parent by setting the bit at position $\rank(\tilde V, \weight(v))$ to one.
To answer a weighted ancestor query $(u,k)$ in $\tilde T$
we compute $i = \rank(\tilde V, k)$, compute the largest $j \ge i$ with $b(u)[j] = 1$
and compute $\select(\tilde V, j)$, from which we can retrieve the identifier of the answer node.
Here the number $j$ is obtained by zeroing out all but the first $i$ least significant bits in $b$ and computing the most significant bit,
which can be computed using multiplication \cite{FredmanW93}.
This concludes the proof.
\end{proof}

\weightedancestors*

\begin{proof}
Initially, we sort the $q$ queries by their weights.
Next, we construct the $w$-decomposition of $T$.
For each leaf $u$ of $T$ belonging to a bottom tree, we store the root $\bottomnode(u)$ of its bottom tree.
For a query $(u,k)$, we first check if $\weight(u')< k$, where $u'$ is the parent of $\bottomnode(u)$.
If this is the case, then the query reduces to a weighted ancestor query in the bottom tree containing $u$. Each bottom
tree is of size $O(w)$, so we we can preprocess all bottom trees in $O(s)$ time and space 
with \Cref{lem:micro-wa} for answering such a query in constant time.

The remaining case is that $k \geq \weight(u')$. Then the query reduces to a query on $u'$.
Observe that $u'$ is an implicit
or explicit node of the top tree. For all such queries, we first issue a weighted ancestor query
on the top tree to find the nearest ancestor $v$ of $u'$ such that $\weight(v)\leq k$. In case when $u'$
is not an explicit node there, we need to access any leaf in the subtree rooted at $u'$ in the top tree.
Such information can be computed and stored together with $\bottomnode(u)$, and then
we can issue the query for the leaf in the top tree instead. All such queries are answered together in
$O(q+s/w)$ total time as explained in~\cite[Lemma~7.2]{KociumakaKRRW20} 
since the top tree has size $O(s/w)$ and we initially sorted the queries.
This gives us, for every such query, a node $v$ of the top tree such that $k\geq \weight(v)$ but for the parent $v'$ of $v$
in the top tree we have $\weight(v') < k$.  Both $v$ and $v'$ are explicit nodes of the top tree and hence also
explicit nodes of $T$. However, in $T$ we are not guaranteed that $v'$ is the parent of $v$. In such a case,
the edge $(v,v')$ of the top tree corresponds to a longer path $v=v_{0}-v_{1}-\ldots v_{\ell+1}=v'$ in $T$, where
$v_{1},v_{2},\ldots,v_{\ell}$ have one child each. Now it remains to find $i$ such that $\weight(v_{i})< k \leq \weight(v_{i+1})$,
i.e. a predecessor query on weights of the nodes on the path. The paths corresponding to different edges of the top
tree are edge-disjoint, hence all such path lengths $\ell$ sum up to at most $s$.
Further, each path is of length at most $w$ by the properties of $w$-decomposition by the following
argument. Consider the inner nodes of the path together with the nodes in all subtrees attached to the
inner nodes and hanging to the left of the path.
Those nodes form a contiguous fragment in the preorder traversal of $T$,
so if there are at least $w$ of them then at least one is chosen. But then, together with any chosen node in
the subtree rooted at the bottom node of the path, this gives us another node of $T'$ among the inner
nodes of the path, a contradiction.
For each edge, we construct and store a separate
predecessor structure implemented with \Cref{lem:predecessor} storing the weights of all nodes on the path.
The overall size and construction time of all those structures is $O(s)$.
Then, each of the remaining queries can be answered by directly in constant time querying the predecessor
structure of the found edge of the top tree.
Thus, the total time is $O(q+s/w)$.
\end{proof}

\section{Substring concatenation}

\label{sec:substringconcat}

The goal of this section is to show \Cref{thm:substringconcat}, i.e. how to preprocess the pattern $p[1..m]$ in $O(m)$ time and space
such that $q$ substring concatenation queries can be answered in time $O(q + m / w)$.
Recall that a substring concatenation query asks,
given two substrings $u$ and $v$ of $p$, check if $uv$ occurs in $p$, and
if so return its occurrence.

We will need two other types of queries: rooted and unrooted LCP queries, see e.g.~\cite{ColeGL04}. Both operate on an arbitrary
compacted trie $T$ storing a set $S$ of suffixes of $p[1..m]$. Given a substring $u = p[i..j]$ of $p$, specified by the pair $(i,j)$,
the rooted LCP query returns the location in $T$ where the search for $u$ starting from the root terminates.
The location is either an explicit node or an implicit node. The unrooted LCP query is additionally given a node
(explicit or implicit) $v\in T$ and returns the location in $T$ where the search for $u$ starting from $v$ terminates.

\begin{lemma}
\label{lem:rootedLCP}
A compacted trie $T$ storing a set $S$ of suffixes of $p[1..m]$ can be preprocessed in $O(|S|)$ time and space
so that $q$ rooted LCP queries can be answered in time $O(q + |S| / w)$ and one call to sorting $q$ integers from $[m]$.
\end{lemma}

\begin{proof}

The preprocessing of $T$ consists of two parts.
First, we traverse the leaves of $T$ in left-to-right order to obtain a sorted list of the ranks
$\{ \isa[i] \mid p[i..m] \in S \}$ of all suffixes in $S$.
Observe that the leaves in $T$ is indeed sorted lexicographically since removing a common prefix of two strings preserves the lexicographical order.
We apply the preprocessing from \Cref{lem:batchedpred} on this list.
Second, we preprocess $T$ with \Cref{thm:weightedancestors}.

To answer a single rooted LCP query for a suffix $p[j..m]$ we find its lexicographical predecessor $p[i..m]$ and successor $p[i'..m]$
among the suffixes in $S$.
Then, we compute the length $\ell$ of the longest common prefix of $p[i..m]$ and $p[j..m]$, and 
the length $r$ of the longest common prefix of $p[i'..m]$ and $p[j..m]$ in constant time.
If $\ell > r$ then the sought node is an ancestor at string depth $\ell$ of the leaf
corresponding to $p[i..m]$, and otherwise it is an ancestor at string depth $r$ of the leaf corresponding
to $p[i'..m]$. %

To answer a batch of $q$ rooted LCP queries concerning substrings $u_{1} = p[i_1..j_1],\ldots,u_{q}=p[i_q..j_q]$, we proceed as follows.
Instead of answering a rooted LCP query for $u_t$, we answer a rooted LCP query for $p[i_{t}..m]$.
If the string depth of the found node is at most $|u_{t}|$ then we return it as the answer, otherwise we need to find
its ancestor at string depth $|u_{t}|$. This can be done with a weighted ancestor query (all such queries are batched).
To compute the predecessor and successor of each $p[j_{t}..m]$ on the sorted list of all suffixes in $S$
we proceed as follows. We sort the ranks $\isa[i_1], \dots, \isa[i_q]$ with one call to sorting
$q$ integers from $[m]$ and issue a batched query to the structure storing a sorted list of ranks of all suffixes in $S$ (\Cref{lem:batchedpred}).
This gives us the predecessor and the successor of $p[i_{t}..m]$ on the sorted list of all suffixes in $S$, for every $t$,
which can be used to obtain the answer to the original rooted LCP query as described earlier with
two lcp queries and a weighted ancestor query (again, all such queries are batched).
\end{proof}

We present the standard reduction from unrooted LCP queries on a trie of size $|S|$
to rooted LCP queries on multiple tries of total size $O(|S| \log |S|)$, see~\cite[Section~5]{ColeGL04}.
We show that the same idea can be used for batched LCP queries.

\begin{lemma}
\label{lem:unrootedLCP}
A compacted trie $T$ storing a set $S$ of suffixes of $p[1..m]$ can be preprocessed in $O(|S|)$ time and space
so that $q$ unrooted LCP queries can be answered in time $O(q + |S| \log |S|/ w)$ and one call to sorting $q$ integers from $[m]$.
\end{lemma}

\begin{proof}
We first define the heavy path decomposition
of an arbitrary tree $T$ on $n$ leaves as follows. For each non-leaf node $u\in T$, we select its child $v\in T$ with the largest
number of leaves in its subtree, and call $v$ the heavy child of $u$, while all other children of $u$ are called light. This decomposes
the nodes of $T$ into node-disjoint paths terminating at leaves, {\em called heavy paths}. The crucial property is that any root-to-leaf path
intersects at most $\log n$ heavy paths. Now consider a compacted trie $T$ storing a set of suffixes $S$ of $p$. We find the heavy
path decomposition of $T$, and note that each heavy path corresponds to a suffix of $p$ (but not necessarily belonging to $S$).
For each node $u\in T$, we create another compacted trie $T_{u}$, called the light subtree of $u$, by extracting the subtree of $u$
but without the edge from $u$ to its heavy child $v$ and the subtree of $v$. In other words, we gather all suffixes corresponding
to the leaves in the subtrees rooted at the light children of $u$, shorten each such suffix by removing the first $d$ characters,
where $d$ is the string depth of $u$, and arrange the truncated suffixes in a compacted trie.
Since a node $v$ is only properly contained in those tries $T_u$ where $u$ has a light child that is an ancestor of $v$
we have $\sum_{u\in T} |T_{u}|=O(|S|\log |S|)$. It is easy to construct all compacted tries $T_{u}$
in $O(|S|\log |S|)$ time by first computing the heavy path decomposition in $O(|S|)$ time, and then extracting the
appropriate subtrees of $T$ in time proportional to their sizes. Now the reduction from unrooted LCP queries to
rooted LCP queries proceeds as follows. We retrieve the heavy path $h$ containing $v$. We compute how far along $h$
we should continue when searching for $u$, this can be done by computing the longest common prefix of two suffixes
of $p$ in constant time. Then, we jump to the last node $v'$ of $h$ that would be visited when searching for $u$, this can be done with
a weighted ancestor query. The latter can be done in $O(q + |S| \log |S|/w)$ time for all $q$ queries
using \Cref{thm:weightedancestors}.
If $v'$ is an implicit node, we are done. Otherwise, we retrieve $T_{v'}$ and issue a rooted
LCP query with the remaining suffix of $u$.
These LCP queries are answered in $O(q + |S| \log |S|/w)$ time using \Cref{lem:rootedLCP}.
\end{proof}

We apply \Cref{lem:unrootedLCP} to a smaller tree on $m/\log m$ leaves which is obtained by decomposing the suffix tree of $p$ with
parameter $x=\log m$. This will allow us to reduce finding the node of the suffix tree corresponding to $uv$
to a rooted LCP query in one of the bottom trees, assuming that $u$ occurs at least $\log m$ times in $p$.
However, we need to design a separate mechanism
for answering queries with $u$ occurring less than $\log m$ times in $p$.

\begin{lemma}
\label{lem:unrootedLCPfew}
The pattern $p[1..m]$ can be preprocessed in $O(m)$ time and space, so that given any substrings $u=u'au''$ and $v$ of $p$,
together with the (explicit or implicit) nodes of the suffix tree of $p^{\R}$ corresponding to $u'$ and $au''$,
and the (explicit or implicit) node of the suffix tree of $p$ corresponding to $v$,
and under an additional assumption that $u''$ is the longest suffix of $u$ that occurs at least $\log m$ times in $p$,
we can check if $uv$ occurs in $p$, and, if so, return its occurrence, in constant time.
\end{lemma}

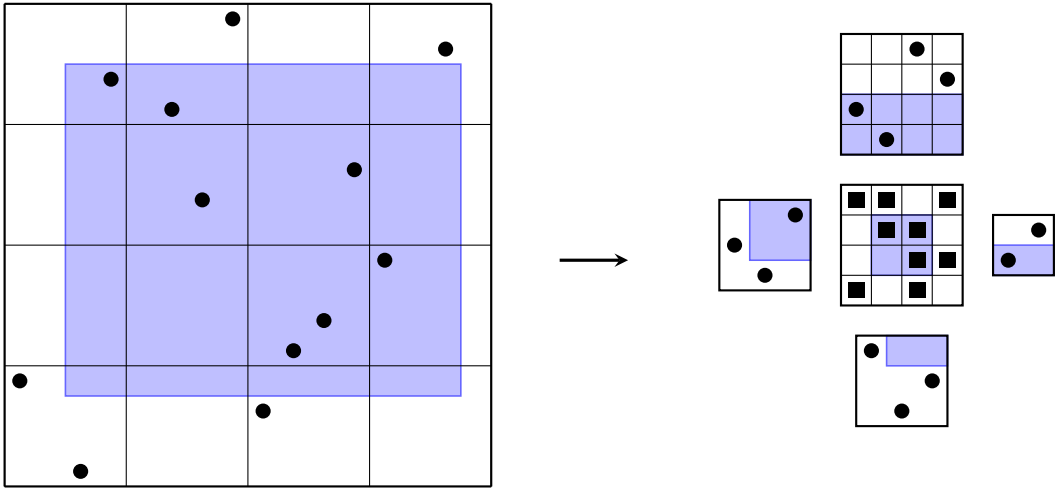
\begin{figure}

\centering

\begin{tikzpicture}

\tikzstyle{rng}=[fill=blue!50, semithick, draw=blue!60, fill opacity = 0.5]

\draw[rng] (0.8,1.2) rectangle (6,5.6);
\draw[rng] (11.4,2.8) rectangle (12.2,3.6);
\draw[rng] (11.0,4.4) rectangle (12.6,5.2);
\draw[rng] (9.8,3.0) rectangle (10.6,3.8);
\draw[rng] (13.0,2.8) rectangle (13.8,3.2);
\draw[rng] (11.6,1.6) rectangle (12.4,2.0);

\tikzstyle{point}=[circle,inner sep = 2pt,fill]

\draw[thick] (0,0) -- (0,6.4);
\draw[thick] (6.4,0) -- (6.4,6.4);
\draw[thick] (6.4,6.4) -- (0,6.4);
\draw[thick] (0,0) -- (6.4,0);
\foreach \x in {0,1.6,3.2,4.8,6.4} {
	\draw (\x,0) -- (\x,6.4);
	\draw (0,\x) -- (6.4,\x);
}
\node[point] at (0.2,1.4) {};
\node[point] at (1.0,0.2) {};
\node[point] at (1.4,5.4) {};
\node[point] at (2.6,3.8) {};
\node[point] at (4.2,2.2) {};
\node[point] at (3.8,1.8) {};
\node[point] at (4.6,4.2) {};
\node[point] at (5.0,3.0) {};
\node[point] at (3.4,1.0) {};
\node[point] at (3.0,6.2) {};
\node[point] at (5.8,5.8) {};
\node[point] at (2.2,5.0) {};

\draw[thick] (11,2.4) -- (11,4.0) -- (12.6,4.0) -- (12.6,2.4) -- cycle;
\foreach \x in {11.4,11.8,12.2} {
	\draw (\x,2.4) -- (\x,4.0);
}
\foreach \x in {2.8,3.2,3.6} {
	\draw (11,\x) -- (12.6,\x);
}
\tikzstyle{bp}=[rectangle,inner sep = 3pt,fill]
\node[bp] at (11.2,2.6) {};
\node[bp] at (12,2.6) {};
\node[bp] at (12,3) {};
\node[bp] at (12.4,3) {};
\node[bp] at (12,3.4) {};
\node[bp] at (11.6,3.4) {};
\node[bp] at (11.2,3.8) {};
\node[bp] at (11.6,3.8) {};
\node[bp] at (12.4,3.8) {};

\draw[->,>=stealth,very thick] (7.3,3) -- (8.2,3);

\draw[thick] (11,4.4) -- (11,6.0) -- (12.6,6.0) -- (12.6,4.4) -- cycle;
\foreach \x in {11.4,11.8,12.2} {
	\draw (\x,4.4) -- (\x,6.0);
}
\foreach \x in {4.8,5.2,5.6} {
	\draw (11,\x) -- (12.6,\x);
}
\node[point] at (11.2,5.0) {};
\node[point] at (11.6,4.6) {};
\node[point] at (12,5.8) {};
\node[point] at (12.4,5.4) {};

\draw[thick] (9.4,2.6) -- (9.4,3.8) -- (10.6,3.8) -- (10.6,2.6) -- cycle;
\node[point] at (9.6,3.2) {};
\node[point] at (10,2.8) {};
\node[point] at (10.4,3.6) {};

\draw[thick] (11.2,0.8) -- (11.2,2) -- (12.4,2) -- (12.4,0.8) -- cycle;
\node[point] at (11.4,1.8) {};
\node[point] at (11.8,1) {};
\node[point] at (12.2,1.4) {};

\draw[thick] (13,2.8) -- (13,3.6) -- (13.8,3.6) -- (13.8,2.8) -- cycle;
\node[point] at (13.2,3.0) {};
\node[point] at (13.6,3.4) {};

\end{tikzpicture}

\caption{Reducing a range emptiness query over $[s] \times [s]$ to queries over $[\sqrt{s}] \times [\sqrt{s}]$.}
\label{fig:grid}

\end{figure}

\begin{proof}
By traversing the suffix tree of $p^{\R}$ we can compute in $O(m)$ time, for each explicit node $u$, the number of leaves
in its subtree, which is equal to the number of occurrences of
its corresponding string in the whole $p$. 
Then, with another traversal we can determine in $O(m)$ time,
for each $j=1,2,\ldots,m$ such that $p[1..j]$ occurs fewer than $\log m$ times in the whole $p$, the largest $i_{j}\leq j$
such that $p[i_{j}..j]$ occurs fewer $\log m$ times in the whole $p$. By definition of $i_{j}$, $p[(i_{j}+1)..j]$
occurs at least $\log m$ in the whole $p$ (possibly, $i_{j}=j$ and then this is the empty string).

We group identical substrings $p[i_{j}..j]$ together, for $j$ such that $i_{j}$ is defined. This can be done
by first
locating their corresponding (explicit or implicit) nodes in the suffix tree of $p^{\R}$ with \Cref{thm:weightedancestors} in $O(m+m/w)$ time,
and then radix-sorting the identifiers of the found nodes in $O(m)$ total time.
By the choice of $i_{j}$,
each group consists of fewer than $\log m$ occurrences, and the overall number of occurrences in all groups
is of course at most $m$. We will construct a separate structure for each such group, and link to it
from the corresponding node. Notice that, while this node might be implicit, any edge of the suffix tree
of $p^{\R}$ contains at most such node (as its ancestor at string depth smaller by 1 must have a larger number
of occurrences, and hence be an explicit node). Therefore, we can store those links in such a way that
later, given an implicit or explicit node of the suffix tree of $p^{\R}$ corresponding to $au''$, we can access
the group built for all occurrences of $au''$ in constant time. It remains to describe how to preprocess
each group in time and space linear in its size so that it can be queried in constant time.

Suppose such a group consists of (identical) substrings $p[i_{j_{k}}..j_{k}]$, for $k=1,2,\ldots,s$.
In preprocessing time we compute the suffix array $\sa^\R[1..m]$ and the inverse suffix array $\isa^\R[1..m]$ with respect to $p^\R$.
We now define a 2-dimensional range emptiness problem and then explain how it is connected to substring concatenation.
For each $k=1,2,\ldots,s$, we create a point $(x,y)$, where $x = \isa[j_k+1]$
and $y = \isa^\R[i_{j_k}-1]$.
This gives us a set $S$ of points from $[m]\times [m]$. Given $v$, we retrieve
the range $[x_{1}..x_{2}]$ (in the rank space) of the leaves of the suffix tree of $p$ in the subtree of the (explicit or implicit) node
corresponding to $v$. Similarly, given $u'$ we retrieve the range $[y_{1}..y_{2}]$ (again, in the rank space) of the leaves of the suffix tree
of $p^\R$ in the subtree of the (explicit or implicit) node corresponding to $u'$.
This can be preprocessed in $O(m)$ time and space for every explicit node of the suffix tree of $p$
and its reversal (and for implicit nodes we access the range of its nearest explicit descendant),
and then retrieved in constant time assuming that the corresponding nodes of both suffix trees
are given. We now observe that to find an occurrence of $u'au''v$ in $p$ we only need to check
if there is a point $(x,y) \in S\cap ([x_{1}..x_{2}]\times [y_{1}..y_{2}])$, and if so retrieve any such point.
From such a point $(x,y)$, we can retrieve the substring occurrence $p[i..j]$ of $au''$ in $p$
by $j = \sa[x]-1$ and $i = \sa^\R[y] + 1$, and thus $uv$ occurs at position $i-|u'|$ in $p$.
It remains to show how to preprocess any set of less than $\log m$ points in linear time and
space for such range emptiness queries.

Consider a set of points $(x_{i},y_{i})\in [m]\times [m]$, for $i=1,2,\ldots,s$ and $s\leq \log m \le w$.
By perturbing the coordinates we can assume that they are pairwise distinct.
We construct and store predecessor structures from \Cref{lem:predecessor} for the $x$
and $y$ coordinates. This allows us to reduce in constant time the query to the rank space,
namely to querying a set of points from $[s]\times [s]$.
Next we will further reduce the grid size to $[\sqrt{s}] \times [\sqrt{s}]$,
see~\Cref{fig:grid}.
We split the $[s]\times[s]$ into
boxes of size $[\sqrt{s}]\times[\sqrt{s}]$. Each horizontal or vertical slice of $\sqrt{s}$
boxes contains at most $\sqrt{s}$ points. Retrieving a point in a rectangle $[x_{1}..x_{2}]\times [y_{1}..y_{2}]$
reduces to retrieving a point in at most two horizontal slices, at most two vertical slices,
and the remaining middle part consisting of complete boxes. For each horizontal and
vertical slice, we again apply reduction to rank space to obtain a set of at most $\sqrt{s}$
points from $[\sqrt{s}]\times [\sqrt{s}]$. The total size of all sets in the new instances is $O(s)$.
For the middle part, we create a set of at most $s$ points from $[\sqrt{s}]\times [\sqrt{s}]$
corresponding to boxes that contain at least one point from $S$. Thus, a query reduces
to constant number of queries on sets of points from $[\sqrt{s}]\times [\sqrt{s}]$, with
the total size of all sets being $O(s)$. We work with an encoding of such a grid
in a single machine word of $s\leq w$ bits obtained by simply concatenating all the rows.
Thus, the set can be stored in a single machine word with a bit set to 1
if and only if the corresponding point belongs to the set. For each $x$ coordinate,
we store a bitmask that allows us to filter points $[x]\times [\sqrt{s}]$, and similarly
for each $y$ coordinate. Together, this allows us to filter points in $[x_{1}..x_{2}]\times [y_{1}..y_{2}]$
with a constant number of standard bitwise operations. This allows us to check in constant
time if the set contains some point from $[x_{1}..x_{2}]\times [y_{1}..y_{2}]$. We can
then retrieve the coordinates of one such point with a constant number of standard bitwise
operations.
\end{proof}

\substringconcat*

\begin{proof}
We start with constructing and storing an $\log m$-decomposition of the suffix tree of $p$. Recall that
the top tree is a compacted trie built for $m/\log m$ suffixes of $p$.
We preprocess the top tree with \Cref{lem:unrootedLCP}, which allows us to answer $q$ unrooted LCP queries
on the top tree in $O(q + m/w)$ time.
Given a batch of substring concatenation queries $(u_1,v_1), \dots, (u_q,v_q)$,
we proceed as follows. First, for every $t$
we locate the node of the suffix tree of $p$ corresponding to $u_t$. This is done using \Cref{thm:weightedancestors}
in $O(q+m/w)$ total time.
For every node $u_t$ which belongs to the top tree we search for $v_t$ in the top tree using unrooted LCP queries
in total time $O(q+m/w)$.
If the (possibly implicit) answer node is at string depth $|u_t v_t|$ we have found the node for the concatenation $u_t v_t$.
Otherwise, we follow the outgoing edge to a bottom tree, labeled by the next character in $v_t$, if any,
and continue searching for a suffix of $v_t$.
To find these outgoing edges efficiently, we store at each node of the suffix tree the first characters on its outgoing
edges in a structure from \Cref{lem:batchedpred}
and sort the $q$ queries by the particular character in $v_t$ in $O(q+m/w)$ time using \Cref{lem:sorting}.

Finally it remains to search $v_t$ from nodes $u_t$ in the bottom tree.
Since the leaves below $u_t$ correspond to the occurrences of $u_t$ in $p$,
it indeed only has at most $\log m$ occurrences.
Let $u_t = u_t' a_t u_t''$ where $u_t''$ is the longest suffix of $u_t$ that occurs at least $\log m$ times in $p$.
To apply \Cref{lem:unrootedLCPfew} we need the nodes of $u_t'$ and $a_t u_t''$ in the suffix tree of $p^\R$
and the node of $v_t$ in the suffix tree of $p$.
First we locate the node of $u_t$ in the suffix tree of $p^\R$
and the node of $v_t$ in the suffix tree of $p$ using two weighted ancestor queries.
We precompute for every node in the suffix tree of $p^\R$, defining some substring $u$,
the length of its longest suffix which occurs at least $\log m$ times in $p$, as in the proof of \Cref{lem:unrootedLCPfew}.
This allows us to locate the desired nodes for $u_t'$ and $a u_t''$ using two weighted ancestor queries,
in total time $O(q + m/w)$ by \Cref{thm:weightedancestors}.
\end{proof}

\bibliographystyle{plain}
\bibliography{refs.bib}

\end{document}